\documentclass[journal]{IEEEtran}
\usepackage{cite}
\usepackage{amsmath,amssymb,amsfonts}
\usepackage{algorithm}
\usepackage{algpseudocode}
\usepackage{pifont}
\usepackage{graphicx}
\usepackage{textcomp}
\usepackage{xcolor}
\usepackage{caption}
\def\BibTeX{{\rm B\kern-.05em{\sc i\kern-.025em b}\kern-.08em
        T\kern-.1667em\lower.7ex\hbox{E}\kern-.125emX}}
\usepackage{balance}

\DeclareMathOperator*{\vc}{vec}
 
\newcommand{\bx}{\mathbf{x}}    
\newcommand{\by}{\mathbf{y}}    
\newcommand{\bs}{\mathbf{s}}    
\newcommand{\bg}{\mathbf{g}}    
\newcommand{\bn}{\mathbf{n}}    
\newcommand{\bh}{\mathbf{h}}    
\newcommand{\bC}{\mathbf{C}}    
\newcommand{\bD}{\mathbf{D}}    
\newcommand{\bI}{\mathbf{I}}    
\newcommand{\bu}{\mathbf{u}}    
\newcommand{\bX}{\mathbf{X}}    
\newcommand{\bY}{\mathbf{Y}}    
\newcommand{\bA}{\mathbf{A}}    
\newcommand{\bB}{\mathbf{B}}    
\newcommand{\bQ}{\mathbf{Q}}    
\newcommand{\define}{\triangleq}    

\newcommand{\mbbC}{\mathbb{C}}	
\newcommand{\mccn}{\mathcal{CN}}	

\newcommand{\bLam}{\mathbf{\Lambda}}

\newcommand{\bSig}{\mathbf{\Sigma}}

\newcommand{\btt}{\bsym{\theta}}

\newcommand{\bzro}{\mathbf{0}}

\newcommand{\ben}{\begin{enumerate}} 	  	
	\newcommand{\een}{\end{enumerate}} 			
\newcommand{\beq}{\begin{equation}} 	  	
	\newcommand{\eeq}{\end{equation}} 			
\newcommand{\bes}{\begin{equation*}}
	\newcommand{\ees}{\end{equation*}}

\newcommand{\bea}{\begin{eqnarray}}		
	\newcommand{\eea}{\end{eqnarray}} 		
\newcommand{\beas}{\begin{eqnarray*}}
	\newcommand{\eeas}{\end{eqnarray*}}
\newcommand{\ba}{\begin{array}}
	\newcommand{\ea}{\end{array}}
\newcommand{\sbea}{\nopagebreak[3]\samepage\begin{eqnarray}}
	\newcommand{\seea}{\end{eqnarray}\pagebreak[0]}
\newcommand{\sbeas}{\nopagebreak[3]\samepage\begin{eqnarray*}}
	\newcommand{\seeas}{\end{eqnarray*}\pagebreak[0]}

\newcommand{\er}[1]{{\rm(\ref{#1})}}

\newcommand{\bit}{\begin{itemize}}
	\newcommand{\eit}{\end{itemize}}
\newcommand{\bsym}{\boldsymbol}

\newcommand{\nn}{\nonumber}

\DeclareMathOperator{\Trace}{Tr}
\DeclareMathOperator{\diag}{diag}

\newtheorem{theorem}{Theorem}
\newtheorem{lemma}{Lemma}

\newtheorem{corollary}{Corollary}
\newenvironment{proof}[1][Proof]{\begin{trivlist}
		\item[\hskip \labelsep {\bfseries #1}]}{\end{trivlist}}

\def\iflatex{\iftrue}
\def\ifcomments{\iffalse}

\begin{document}
\title{Impedance Variation Detection at MISO Receivers}
\author{\IEEEauthorblockN{Shaohan Wu}\\
  \IEEEauthorblockA{
    {MediaTek USA Inc.}, Irvine, CA 92620 \\
    {Shaohan.Wu}@mediatek.com}
}
  \maketitle
  \begin{abstract}
  Techniques have been proposed to  estimate unknown antenna impedance due to time-varying near-field loading conditions at multiple-input single-output (MISO) receivers. However, it remains unclear when a change occurs and impedance estimation becomes necessary.  In this letter, we address this problem by formulating it as a hypothesis test. Our contributions include deriving a generalized likelihood-ratio test (GLRT) detector to decide if the antenna impedance has changed  over two groups of packets. This GLRT formulation leads to a novel optimization problem, but we propose a binary search based algorithm to solve it efficiently. Our derived GLRT detector enjoys a better detection and false alarm trade-off when compared with a well-known, reference detector in simulations.  As one result, more transmit diversity significantly improves detection accuracy at a given false alarm rate, especially in slow fading channels.
  \end{abstract}
  

\section{Introduction}
Impedance matching between the receive antenna and front-end can significantly impact channel capacity in wireless, multi-path channels \cite{domi2}. To implement capacity-optimal matching, receivers must know the antenna impedance. Due to time-varying near-field loading conditions, antenna impedance may change significantly.  Researchers have proposed techniques to estimate the unknown antenna impedance in real-time \cite{wu,wu2,hass,wu_mimo_hybrid,wu_pca_mimo,wu_pca_miso}. 
However, it remains unclear when these estimation algorithms should be triggered. In this letter, we propose a mechanism to detect impedance change. 

Hassan and Wittneben considered joint channel and impedance
estimation for multiple-input multiple-output (MIMO) receivers. In their work, they vary the receiver load impedance and apply a least-square estimation technique to jointly estimate the channel and the antenna impedance under uncorrelated, fast-fading channels \cite{hass}. Wu and Hughes  considered impedance estimation at MISO receivers using a hybrid estimation framework, where channel and impedance are estimated jointly \cite{wu2}. Wu extended it to MIMO \cite{wu_mimo_hybrid}. The drawback of this approach is the impedance estimator is generally biased. Consequently, Wu and Hughes derived a closed-form maximum-likelihood (ML) estimator for impedance, treating the channel gains as nuisance parameters, for both MISO and MIMO receivers under i.i.d. Rayleigh fading \cite{wu_pca_miso,wu_pca_mimo}. They also demonstrated that these ML estimators are efficient with sufficient signal-to-noise-ratio (SNR) and diversity. 
Since extra circuitry operations are needed to apply these estimation algorithms, ideally we want to apply them only when impedance has changed for battery-life sensitive devices like mobile phones and smart watches. However, such impedance variation detection algorithms remain unknown.

In this letter, we address this problem by formulating it as a hypothesis test. Our contributions include deriving a generalized likelihood-ratio test (GLRT) detector to decide if the antenna impedance has changed  over two groups of packets. Our derived GLRT detector is generally not in a closed-form due to non-linear nature of the likelihood function. This formulation leads to a novel optimization problem. But we propose a binary-search based algorithm, which solves said problem exactly and efficiently. Comparison between our proposed GLRT detector to a well-known, reference detector suggests our GLRT detector generally enjoys a sizable benefit in detection and false alarm trade-off. Importantly, our proposed detector is able to prolong battery life for smart wearable devices with adaptive impedance matching capabilities. 

In this letter,  we use $Z_A$ and $Z_L$ to represent the antenna and load impedance, respectively. Bold lower case notations are vectors, while bold upper case ones are matrices. Note $\bh$ is the channel coefficients needed by communication algorithms and $\sigma_{h}^2$ is its variance. Also,  $(\cdot)^T$ is transpose and $(\cdot)^H$ is the Hermitian operator. 

The rest of this paper is organized as follows. We present the system model and formulate this problem using hypothesis testing in Sec.~\ref{sec_model}. Then in  Sec.~\ref{sec_det} we derive maximum likelihood (ML) estimators and consequent the GLRT detector. The performance of this detector is  explored via numerical examples in Sec.~\ref{sec_sim}. We conclude and point out a future direction in Sec.~\ref{sec_conclusion}.

\section{Problem Formulation}\label{sec_model}\label{sec_test}
Suppose there is a single antenna at the receiver. 
The transmitter sends a known training sequence, $\bx_1 , \ldots , \bx_T\in\mbbC^N$. If channel gains $\bh$ are fixed for the duration of training, the received observations can be compactly written by \cite[eq.~1]{vu}
\beq\label{eq_u_vec}
\bu  \ = \   \bh^T\bX + \bn_L \ , 
\eeq
where we define $\bX \define [ \bx_1,\cdots,\bx_T ]\in\mbbC^{N\times T}$ and $\bn_{L} \sim\mathcal{CN}(0, \sigma_L^2\bI_T)$ is independent and identically distributed.

We assume multiple transmit antennas exist at the base station, i.e., $N>1$, and they are sufficiently separated such that $\bh \sim\mccn\left(\bzro_N,\sigma_h^2\bI_N\right)$.
We assume block-fading, and consequently, for each data packet, a training sequence precedes the data sequence \cite{bigu}. 
We assume $\bX$ is orthogonal and equal-energy, i.e.,  
\beq\label{eq_X}
\bX\bX^H \ = \ \frac{PT}{N}\bI_N \ .
\eeq
It can be shown that  a sufficient statistic exists which summarizes information in $\bu$ \er{eq_u_vec} and we denote it as \cite{kay_est}
\beq\label{eq_y}
\by \ \define \  \frac{N}{PT}\left(\bu\bX^H\right)^T \ = \ \bh + \bn \ ,
\eeq
where $\by$ and $\bn$ are in $\mbbC^N$, $\bn\sim\mccn(\bzro,\sigma_n^2\bI_N)$, and 
\beq\label{eq_sigma_n}
\sigma_n^2 \ = \ \frac{N \sigma_L^2}{T P} \ .
\eeq
 
The goal of this letter is to derive detectors for antenna impedance variations over time. We observe any such variation manifests itself as a change to the variance $\sigma_h^2$. Therefore, the goal can be equivalently stated as detecting changes to the variance of the fading channel over time.



To state this problem clearly, we reiterate the goal is to detect changes in antenna impedance due to time-varying loading conditions in the near-field. We assume this loading condition varies at a rate much slower than the duration of a packet\footnote{For example, movements of human users serve a common cause to time-varying near-field loading conditions. Its rate can be modeled on the order of seconds. However, a packet is often on the order of a millisecond. }.  To this end, we assume the antenna impedance remains the same in each group of packets, and may vary from group to group. This serves a reasonable approximation of the practical phenomena that we want to model, and leads to tractable solutions shown next. 
 
Consider two groups of observation each with $L$ consecutive packets. In general, temporal correlation exists between these packets. Also assume the two groups are sufficiently apart in time, such that any packet in either group is statistically independent to any packet in the other. From \er{eq_y}, we have, 
\beq\label{eq_Ys}
\bY_i \ = \ \begin{bmatrix}
    \by_{i, 1} & \cdots & \by_{i, L}
\end{bmatrix} \ ,
\eeq
with  $\by_{i,k} = \bh_{i,k} + \bn_{i,k}$, where $i=1,2$ is group index and $k=1,\dots, L$  packet index. Note  $\bh_{i,k}$ depends on the unknown antenna impedance $Z_{A,i}$ via a voltage divider \cite[eq.~4]{wu},  
\beq\label{eq_h}
\bh_{i,k} \ \define \ \frac{Z_L}{Z_{A,i} + Z_L}\bg_{i,k} \ , 
\eeq
where $\bg_{i,k}$ are path gains of the physical propagation channel and $Z_L$ is the load impedance\cite[Fig.~1]{wu2}. 

The channel and noise vectors are zero-mean complex Gaussian, and with \er{eq_sigma_n} and \er{eq_h}, we can write for all $k$, 
\beq\label{eq_cov}
\by_{i,k} \sim\mccn\left(\bzro_N, \left(\sigma_{h,i}^2+\sigma_n^2\right)\bI_N\right)  
\ ,
\eeq
where $i=1,2$, unknown variances of path gains $\bg_{i,k}$ are denoted as $\sigma_{g,i}^2$, and we have from \er{eq_h}, 
\beq\label{eq_sgmh}
\sigma_{h,i}^2 \ = \ \left|\frac{Z_L}{Z_{A,i}+Z_L}\right|^2 \sigma_{g,i}^2\ . 
\eeq
The spatial independence of $\bg_{i,k}$  comes from uncoupled transmit antennas. 

We formulate the problem of antenna impedance variation detection as a hypothesis test,
\bea\label{eq_hyp_test}
\begin{cases}
    H_0:& \sigma_{h,1}^2 \ = \ \sigma_{h,2}^2 \ ,\\
    H_1:& \sigma_{h,1}^2 \ \neq \ \sigma_{h,2}^2 \ ,
\end{cases}
\eea
with the assumption that inequality in the channel variance is solely caused by inequality in antenna impedance \eqref{eq_sgmh}. 

If the null hypothesis $H_0$ is true, then the antenna impedance remains identical in the two groups considered. Otherwise, a change in antenna impedance has occurred and antenna impedance estimation becomes necessary. 

In the next section, we present the main results of this letter, i.e., the generalized likelihood-ratio test (GLRT) detector for the problem formulated in \er{eq_hyp_test}. 

\section{Main Results}\label{sec_det}

We consider the observations $\bY_1$ and $\bY_2$ in \er{eq_Ys}. Firstly, they are statistically independent. Also, the distribution of $\bY_i$ with $i=1,2$  is, 
\beq\label{eq_Yi_dist}
\vc\left(\bY_i\right) \sim\mccn \left(\bzro_{NL}, \sigma_{h,i}^2\bC\otimes\bI_N+\sigma_n^2\bI_{NL}\right) \ ,
\eeq
where $\bC$ is the known channel temporal correlation matrix but channel variance $\sigma_{h,i}^2$'s are unknown. We define
\beq\label{eq_theta_vec}
\btt \ \define \ \begin{bmatrix}
    \theta_1 & \theta_2
\end{bmatrix}^T\ = \ \begin{bmatrix}
    \sigma_{h,1}^2 & \sigma_{h,2}^2
\end{bmatrix}^T \ . 
\eeq

To test \er{eq_hyp_test}, we need the following sufficient statistic, 
\beq\label{eq_s_def}
\bs \ \define \ \begin{bmatrix}
    \bs_1 & \bs_2
\end{bmatrix}^T\ ,
\eeq
where 
for $i=1,2,$
\beq\label{eq_ss}
\bs_i \ \define \ \diag\left(\frac{1}{N}\bQ^H\bY_i^T\bY_i^*\bQ\right) \ , 
\eeq
where $\bQ$ is from the eigen-decomposition of $\bC$,  
\beq\label{eq_Sigma_def}
\bC \ = \ \bQ\bLam\bQ^H \ ,~~\bLam \ \define \ \diag(\lambda_1,\dots,\lambda_L) \ . 
\eeq
Note $\bs_i$ in \er{eq_ss} is the sample covariance after decorrelation. 

Note \er{eq_hyp_test} can be written in terms of $\btt$ defined in \er{eq_theta_vec}. 
When the null hypothesis is true, we estimate $\theta_1$ using $\bs$. But under the alternative hypothesis, we estimate $\theta_1$ and $\theta_2$ using $\bs_1$ and $\bs_2$, respectively. 

By definition, the generalized likelihood-ratio test (GLRT) detector plugs in the maximum-likelihood (ML) estimator back into its likelihood function \cite[Sec.~6.4.2]{kay}. 
Thus, we present the ML estimator of the unknown channel variance first, and then derive GLRT detectors. 

\subsection{ML Estimators}

Without loss of generality, we firstly discuss how to estimate the unknown channel variance of either group specified in \er{eq_Ys}. 
With a slight abuse of notation, we label the unknown parameter that we want to estimate as 
\beq\label{eq_theta}
\theta \ \define \ \sigma_h^2 \ , 
\eeq
and consider a statistic $\bs$ that is either $\bs_i$ defined in \er{eq_ss}. 

As shown in \er{eq_pdf_s} in the Appendix, we can write the log-likelihood function (LLF) in terms of $\theta$, 
\bea\label{eq_llf_theta}
&&\mathcal{L}(\theta;\bs) \ \define \ \ln p(\bs;\theta) \nn\\
& = & c- N\sum_{k=1}^{L}\left[\ln\left(\lambda_k\theta+\sigma_n^2\right)+\frac{s_k}{\lambda_k\theta+\sigma_n^2}\right] \ ,
\eea
where $s_k$ is the $k$-th element of $\bs$, both $\lambda_k$ and $\sigma_n^2$ are known, and $c$ is a constant independent of $\theta$. 

We observe several important differences between our derived LLF in \er{eq_llf_theta} and a representative reference \cite[eq.~9.6]{kay}. Note $\theta$ only exists in our formulation and the $\lambda_k$'s make our problem more challenging to solve. Thus, our formulation in \er{eq_llf_theta} is fundamentally different from \cite[eq.~9.6]{kay}. Further, we propose an efficient algorithm to estimate $\theta$ in Algorithm \ref{algo_line_srch}. 

By definition, the ML estimator (MLE) is the maximizer of the LLF, i.e.,
\beq\label{eq_mle}
\hat{\theta}_{ML} \ \define \ \arg\max_\theta \mathcal{L}(\theta) \ .
\eeq

The gradient of the LLF \er{eq_llf_theta} is readily derived, i.e., 
\bea\label{eq_grad_hess}
\nabla\mathcal{L}(\theta) & = & \sum_{k=1}^{L}  \left[\frac{s_k\lambda_k}{\left(\lambda_k\theta+\sigma_n^2\right)^2} -  \frac{\lambda_k}{\lambda_k\theta+\sigma_n^2} \right] \ .
\eea

A necessary condition for $\hat{\theta}_{ML} $ is the gradient vanishes. 
However from \er{eq_grad_hess}, a closed-form expression generally does not exist, due to summation of non-linear functions of $\theta$.  
But first, we consider a special case where the MLE is in closed-form. 

\subsubsection{ML Estimator under i.i.d. Fading}
In this special case, we have, $\lambda_k = 1$ for all $k$,. 
The MLE of $\theta$ exists in closed-form based on \er{eq_grad_hess}, i.e.,
\beq\label{eq_mle_iid}
\hat{\theta}_{ML,iid} \ = \ \max 
\left\{\frac{\sum_{k=1}^{L}s_k}{L} - \sigma_{n}^2, 0 \right\} \ . 
\eeq
Note the $\max\left\{\cdot\right\}$ operator is necessary as $\theta$ is variance and hence non-negative by definition. 

\subsubsection{ML Estimator in general}

Since $\theta\in[0,\infty)$, we could set the upper bound a practically large number. Then MLE falling outside the practical range is negligible. We could use a simple line search algorithms with a fixed step size. Thus, we are guaranteed to find the MLE, but a disadvantage is its computational inefficiency. 

The gradient in \er{eq_grad_hess} is positive when $\theta =0$  and  crosses zero only once and then remain negative. This renders the procedure of finding the MLE simple. We propose a binary-search based algorithm to solve the MLE efficiently. To be specific, we start with a range from zero to a sufficiently large number and use binary search to iteratively narrow down this range until it is in close proximity of the MLE. In each iteration, we replace one of the two bounds of this range by its middle point if the gradient of this bound and the middle point share same sign. 

\begin{algorithm}
    \caption{Binary-Search Based Algorithm}
    \label{algo_line_srch}
    \begin{algorithmic}[1]
        \State Let $\theta_l=0$, $\theta_u = \alpha\sum_{k=1}^{L}s_k$, $\epsilon=10^{-3}$
        \While{$\nabla\mathcal{L}(\theta_l)\cdot\nabla\mathcal{L}(\theta_u)<0$}
        \State $\theta_m={(\theta_l+\theta_u)}/{2}$
        \If {$\left|\nabla\mathcal{L}(\theta_m)\right|<\epsilon$} break
        \EndIf
        \If {$\nabla\mathcal{L}(\theta_l)\cdot\nabla\mathcal{L}(\theta_m)>0$} $\theta_l=\theta_m$
        \Else \ $\theta_u=\theta_m$
        \EndIf
        \EndWhile
        \State $\hat{\theta}_{ML}=\theta_m$
    \end{algorithmic}
\end{algorithm}
Note $\alpha=10$ (by default) and $\epsilon=10^{-3}$ are parameters of Algorithm \ref{algo_line_srch} to choose the trade-off between accuracy and speed, i.e., number of iterations needed. 

\subsection{GLRT Detectors}
Classical detection theory offers many important results on hypothesis testing \cite{vt,kay}. 
In general, we can write down the likelihood ratio test \cite[eq.~2.195]{vt},
\bea\label{eq_test}
t\left(\bs\right) \ = \ \mathcal{L}\left(\bs| H_1\right) - \mathcal{L}\left(\bs|H_0\right)  \ \mathop{\gtreqless}^{H_1}_{H_0} \ \gamma \ , 
\eea
where $\mathcal{L}(\bs|H_i)$ denotes the log-likelihood function (LLF) in terms of sufficient statistic $\bs$ in \er{eq_s_def} conditioned on a hypothesis $H_i$ with $i=0,1$ as defined in \er{eq_hyp_test}, and $\gamma$ denotes the threshold for this test that we get to choose. 

\begin{lemma}[LLF]\label{lem_llf}
    Given the observations $\bY_1$ and $\bY_2$ in \er{eq_Ys},  it can be shown that its log-likelihood function (LLF) can be equivalently written in terms of $\bs$, i.e.,
    \bea\label{eq_llf}
    &&\mathcal{L}(\btt) \ \define \ \ln p(\bY_1, \bY_2;\btt)\\
    &=& c- N\sum_{i=1}^{2}\sum_{k=1}^{L}\left[\ln\left(\lambda_k\theta_i+\sigma_n^2\right)+\frac{s_{i,k}}{\lambda_k\theta_i+\sigma_n^2}\right] \ ,\nn
    \eea
    where $s_{i,k}$ is the $k$-th element of $\bs_i$ in \er{eq_ss} and $\lambda_k$ in \er{eq_Sigma_def}. 
\end{lemma}
The proof of Lemma \ref{lem_llf} is straightforward based on derivations  in finding \er{eq_llf_theta}, and thus omitted here due to page limitation. 

Next we derive the generalized likelihood ratio test (GLRT) for the problem \er{eq_hyp_test} using ML estimates to replace unknown parameters in \er{eq_test}. 
The GLRT detector is stated in a theorem. 

\begin{theorem}[GLRT Detector]\label{th_glrt}
  Suppose the original observation $\bY_1$ and $\bY_2$ and the sufficient statistic are given in \er{eq_Ys} and \er{eq_ss}, respectively. Also suppose the channel variances as defined in \er{eq_theta_vec} are unknown and to be estimated. The generalized likelihood ratio test (GLRT) for the problem in \eqref{eq_hyp_test} is \cite[eq.~6.12]{kay},
  \bea\label{eq_glrt}
  t_{GLRT} &\define & \max_{\btt}  \mathcal{L}(\btt|H_1) - \max_{\btt} \mathcal{L}(\btt|H_0) \\
  &=& N\sum_{k=1}^{L}\left[2\ln\left(\lambda_k\tilde{\theta}_1+\sigma_n^2\right)+\sum_{i=1}^{2}\frac{s_{i,k}}{\lambda_k\tilde{\theta}_1+\sigma_n^2}\right] \nn\\
  &&-N\sum_{i=1}^{2}\sum_{k=1}^{L}\left[\ln\left(\lambda_k\hat{\theta}_i+\sigma_n^2\right)+\frac{s_{i,k}}{\lambda_k\hat{\theta}_i+\sigma_n^2}\right] \nn\ ,
  \eea
  where $\hat{\theta}_i$ for $i=1,2$ is the MLE \er{eq_mle} when $H_1$ is true \er{eq_hyp_test}, i.e, $\theta_1\neq\theta_2$, while $\tilde{\theta}_1$ is that when $H_0$ is true, i.e., $\theta_1=\theta_2$, $s_{i,k}$ is the $k$-th element of $\bs_i$ defined in \er{eq_ss} and $\lambda_k$ in \er{eq_Sigma_def}. 
\end{theorem}

\begin{proof}
  The proof is straightforward by plugging the respective ML estimators \er{eq_mle} conditioned on either hypothesis in \er{eq_hyp_test} into their corresponding LLF in \er{eq_llf}. 
\end{proof}

In the next result, we consider a special case where the fading channels are i.i.d..

\begin{corollary}[GLRT Detector in i.i.d. Fading]
  If fading channel is i.i.d. in time, i.e., $\bC=\bI_L$ and equivalently $\lambda_k=1$ for all $k=1,\dots, L$, then the well-known Bartlett's test is the GLRT detector  \cite{ciuo}, i.e., 
  \beq\label{eq_bt}
  t_{BT} \ = \ \ln \left[\frac{(s_{1}+s_{2})^2}{4s_{1}s_{2}} \right]\ , 
  \eeq
  where for $i=1,2$, we have defined
  \beq
  s_i \ \define \ \frac{1}{L}\sum_{k=1}^{L}s_{i,k} \ ,
  \eeq
  and $s_{i,k}$ is the $k$th entry in $\bs_i$ in $\bs$ as defined in \er{eq_s_def}. 
\end{corollary}
 
 \begin{proof}
   This proof is straightforward based on the i.i.d. MLE derived in \er{eq_mle_iid}. Under the null hypothesis, we have 
   \beq
   \max_{\btt}\mathcal{L}(\btt|H_0) \ = \ -2NL\left[\ln\left(\frac{s_1+s_2}{2}\right)+1\right] \ ,
   \eeq
   and under the alternative hypothesis, 
   \beq
   \max_{\btt}\mathcal{L}(\btt|H_1) \ = \  -NL\left[\ln\left(s_1\right)+1 +  \ln\left(s_2\right)+1\right]\ . 
   \eeq
   Plugging above two formulas into \er{eq_glrt}  and some equivalent simplification result in \er{eq_bt}. 
 \end{proof}
 
 The well-known Bartlett's test in \er{eq_bt} is the GLRT detector when fading channel is i.i.d; otherwise it is sub-optimal and serves as a reference for the GLRT detector in \er{eq_glrt}. 
 
 We compare complexity of both derived detectors in \er{eq_glrt} and \er{eq_bt} against the best single-antenna impedance estimator, i.e., the ML estimation in \cite[eq. 30]{wu_pca_miso}. The results of complexity averaged over each packet are summarized in Table \ref{tab_complexity}. 
 
 \begin{table}[h!]
     \vspace{-4pt}
   \begin{center}
       \caption{Algorithm Complexity Comparison}
   \begin{tabular}{|c | c | c | c | c|} 
     \hline
     Algorithm & $+$ & $\times$ & $\div$ & Circuitry Op. \\  
     \hline
     GLRT \er{eq_glrt} & 4.5$N_{itr}$ & 3$N_{itr}$ & 3$N_{itr}$ & $P_{act}/2L$ \\ 
     \hline
     Bartlett \er{eq_bt} & 1 & $1/L$ & $1/L$ & $P_{act}/2L$ \\
     \hline
     MLE \cite[(30)]{wu_pca_miso} & 5 & 6 & 2  & 1\\
     \hline
   \end{tabular}
   \label{tab_complexity}
   \vspace{-6pt}
\end{center}
 \end{table}

 Note in Table \ref{tab_complexity}, $+$, $\times$, and $\div$ represent real addition, multiplication, and division, receptively. Also, $N_{itr}$ is the number of  iterations of Algorithm \ref{algo_line_srch} needed to find  $\hat{\theta}_{ML}$ in \er{eq_mle}.  Circuitry operations include switching load impedance $Z_L$ during training sequence  to estimate  $Z_A$ and tuning $Z_L$ to match to estimates of $Z_A$ at packet rate. These circuitry operations consume much more power than the basic operations. Thus, a more energy-efficient solution is one less likely to \textit{act} on these circuitry operations.  Generally speaking, $P_{act}$ is the probability of deciding the alternative hypothesis \er{eq_hyp_test}, 
 \beq\label{eq_P_act}
 P_{act} \ = \ P(t>\gamma|H_0)P(H_0) + P(t>\gamma|H_1)P(H_1) \ ,
 \eeq 
 where $t$ could be any detector and $\gamma$ is its associated threshold  \er{eq_test}. 
 We will compare detectors \er{eq_glrt} and \er{eq_bt}'s  false alarm rate $P(t>\gamma|H_0)$ given a fixed detection rate $P(t>\gamma|H_1)$. 
 
 In the next section, we evaluate the performance of our derived detectors via numerical simulations. 
\section{Numerical Results}\label{sec_sim}
In this section, we explore the performance of  the detectors derived in the previous section through numerical examples. Consider a narrow-band MISO communications system with $N=4$ transmit antennas and a single receive antenna, whose carrier frequency is 900 MHz. The duration of each data packet equals to a sub-frame of 5G New Radio (NR), i.e., $T_s=1$ ms. Block-fading channel is assumed, such that during one data packet, the channel remains the same, but it generally varies from packet to packet \cite{bigu}. 
We take the training sequence $\bX$ in \er{eq_X} from a normalized discrete Fourier transform (DFT) matrix of dimension $T=64$, e.g., \cite[eq. 10]{bigu}. Specifically, $\bX$ is chosen as the first $N$ rows of this DFT matrix. Also, the SNR is 10dB. For parameters used in Algorithm \ref{algo_line_srch}, we choose $\alpha=10$ and $\epsilon=0.001$. 

Before a change occurs, the antenna impedance is that of a dipole, i.e. $Z_{A,1} = 72 + 42j \,\Omega$.  If the near-field loading conditions change, the antenna impedance changes accordingly. We assume the changed antenna impedance is $Z_{A, 2} = 72 + 100j \,\Omega$. 
The load impedance is fixed at $Z_L=50 \,\Omega$, as we assume no knowledge of $Z_{A,i}$ with $i=1,2$. Two groups of packets are collected with each group consisting of $L=20$ packets. Antenna impedance in the first group is always $Z_{A,1}$. In the second group, it remains $Z_{A,1}$ if the null hypothesis is true; otherwise, it changes to $Z_{A, 2}$ under  alternative hypothesis \er{eq_hyp_test}. 

The average post-detection SNR of a received symbol is defined as \cite[Sec. VIII]{bigu}, using \er{eq_u_vec} and \er{eq_X}, 
\beq\label{eq_symbolSNR}
\rho \ \define \ \frac{E\left[\Trace\left[\bh^T\bX\bX^H\bh^*\right]\right]}{T\cdot\sigma_{L}^2} \ = \ \frac{P\sigma_{h,i}^2}{\sigma_{L}^2}\ ,
\eeq
where $\sigma_{L}^2$ is the noise variance at the receiver and $\sigma_{h,i}^2$ is the channel variance in \er{eq_sgmh} calculated with $Z_{A,i}$ for $i=1,2$, and the expectation is over randomness of channel. 

\begin{figure}[t!]
    \vspace{-12pt}
    \begin{center}
        \includegraphics[width=.45\textwidth, keepaspectratio=true]{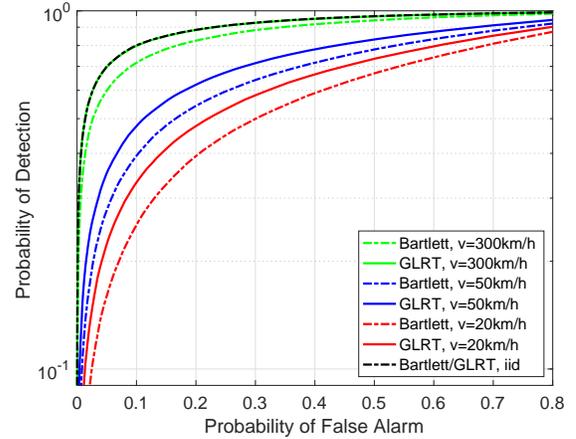}
    \end{center}
    \vspace{-6pt}
    \caption{ROC in Correlated Rayleigh Fading, $L=20$}
    \vspace{-12pt}
    \label{fig_ROC_corr_CH}
\end{figure}

Clarke's model is assumed for temporal correlation \cite{zhen}, where the correlation is
  $J_0(2\pi f_dT_s |l|)$, 
$J_0(\cdot)$ is the zeroth-order Bessel function of the first kind, $T_s=1\,$ms  is the sampling interval, and $l$ is the sample difference. The maximum Doppler frequency is 
$f_d\define v/\lambda$, 
where $v$ is the velocity  of the fasting moving scatterer and $\lambda$ the wave-length of the carrier frequency. Three velocities are chosen in our simulations, i.e., $v=300, 50, 20$ km/h, which represent a user on a high-speed train, driving on a local road, and walking in urban area, respectively.

We can see in Fig. \ref{fig_ROC_corr_CH},  channel correlation significantly impacts the trade-off between probability of detection (PD) and probability of false alarm (PFA). Specifically, for the case with very high speed ($v=300$ km/h), the receiver operating characteristic (ROC) of the GLRT detector is indistinguishable from its counterpart under i.i.d. fading. But the sub-optimal Bartlett's test is a sizable gap away. The average number of iterations $N_{itr}$ in Table \ref{tab_complexity} is around 13.6. Furthermore, for all velocities considered and with a fixed PD at $80\%$, the GLRT detector leads to a PFA at least $5\%$ smaller than using the reference Bartlett's test. This means our proposed GLRT detector could be more energy-efficient due to \er{eq_P_act} and complexity analysis in Table \ref{tab_complexity}. 

We also consider the impact of transmit diversity on the ROC. Since massive MIMO is widely discussed and used in 5G and beyond, we assume the number of transmit antennas takes value $N=4, 16, 64$. To ensure fair comparison, same total transmit power is assumed. This means the numerator at the right-hand-side of \er{eq_symbolSNR} remains the same for all values of $N$. 
In Fig.~\ref{fig_ROC_TxDiv_corr_CH}, we observe a noticeable benefit for our proposed GLRT detector over the reference Bartlett's test for all $N$ under slow fading with  $v=20$ km/h. In particular, this benefit is around 10\% reduction in PFA for a fixed PD at $80\%$ with $N=4,16$. This means our derived GLRT detector is more power efficient than its Bartlett counterpart, despite its slightly more computation ($N_{itr}=13.6$ in Table \ref{tab_complexity}). If $N = 64$, ROC's of  GLRT and Bartlett detectors approach the perfect classification at the upper left vertex, with the GLRT slightly better than the Bartlett. Thus, higher transmit diversity order generally leads to improved trade-off between PD and PFA.

\begin{figure}[t!]
    \vspace{-12pt}
    \begin{center}
        \includegraphics[width=.45\textwidth, keepaspectratio=true]{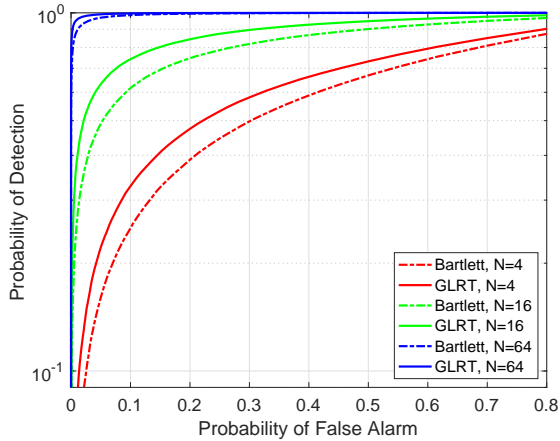}
    \end{center}
    \vspace{-6pt}
    \caption{Impact of Transmit Diversity of ROC, $L=20$}
    \vspace{-12pt}
    \label{fig_ROC_TxDiv_corr_CH}
\end{figure}

\section{Conclusion}\label{sec_conclusion}
In this letter, we formulated the problem of  antenna impedance variations detection at MISO receivers as a hypothesis test on equality of the variance of two complex Gaussian random variables. We derived the GLRT detector for this test, which is not in a closed-form in general. We then proposed a binary-search based algorithm to find the GLRT detector efficiently. When the fading channel is i.i.d., the well-known Bartlett's test is proven as the GLRT detector; but when fading channel is correlated in time, this Bartlett's test is sub-optimal and serves as a reference to our derived GLRT detector. 
Our GLRT detector enjoys a better detection and false alarm trade-off when compared to the reference Bartlett's test, and consequently could save more power for smart wearable devices. 
Numerical simulations also suggest higher transmit diversity significantly improves the detection rate at a given probability of false alarm, especially in slow fading channels. 

One potential future direction is to consider transmit correlation and/or a Rician channel. Another future direction is to identify more computationally efficient detectors with similar ROC behavior as our derived GLRT detector. 

\appendix

The probability density function (PDF) of $\by$ is
\bea\label{eq_pdf_y}
p(\by) 
&= & \det\left(\pi\bSig_\by\right)^{-1}\exp\left(-\by^H\bSig_\by^{-1}\by\right)\\
&=& \pi^{-NL}\det\left(\sigma_h^2\bC+\sigma_n^2\bI_{L}\right)^{-N}\nn\\
&&\exp\left(-\Trace\left[\bY^*\left(\sigma_h^2\bC+\sigma_n^2\bI_{L}\right)^{-1}\bY^T\right]\right) \ ,\nn
\eea
where 
\beq
\bSig_\by \ = \ \left(\sigma_h^2\bC+\sigma_n^2\bI_{L}\right)\otimes\bI_N \ , \nn
\eeq
and the following identity is used \cite[eq.~2.116]{hjor},
\beq
\Trace\left[\bA\bB\bC\bD\right] \ = \left[\vc\left(\bB\right)\right]^T\left[\bC\otimes\bA^T\right]\vc\left(\bD^T\right) \ . \nn
\eeq

Without loss of generality, let $\bs$ be either $\bs_i$ in \er{eq_ss}, $i=1,2$. The PDF in \er{eq_pdf_y} can be simplified into, 
\bea\label{eq_pdf_s}
p(\bs)  & = & \pi^{-NL}\det\left(\sigma_h^2\bLam+\sigma_n^2\bI_{L}\right)^{-N}\nn\\
&&\exp\left(-\Trace\left[\left(\sigma_h^2\bLam+\sigma_n^2\bI_{L}\right)^{-1}\left(\bQ^H\bY^T\bY^*\bQ\right)\right]\right) \nn\\
&=& \pi^{-NL} \left(\prod_{k=1}^{L} \beta_k^{-N}\right)  \cdot \exp\left(-N\sum_{k=1}^{L}\frac{s_k}{\beta_k}\right) \ ,
\eea
where $s_k$ is the $k$-th element of $\bs$ and we define $\beta_k=\lambda_k\sigma_h^2+\sigma_n^2$, for $k=1,\dots,L$. 

\bibliographystyle{unsrt}

\end{document}